\documentclass{article}

\usepackage[letterpaper,top=2cm,bottom=2cm,left=3cm,right=3cm,marginparwidth=1.75cm]{geometry}

\usepackage[utf8]{inputenc}
\usepackage[T1]{fontenc}    
\usepackage{url}            
\usepackage{booktabs}       
\usepackage{amsfonts}       
\usepackage{nicefrac}       
\usepackage{microtype}  
\usepackage{lipsum}
\usepackage{amsthm}
\usepackage{amsmath}
\usepackage{comment}
\usepackage{amssymb}
\usepackage{paralist}
\usepackage{xcolor}
\usepackage{dsfont}
\usepackage{algorithmic}
\usepackage{algorithm}

\usepackage{tikz}
\usepackage{graphics}
\usetikzlibrary{trees}
\usetikzlibrary{automata,positioning}
\usepackage{array}

\newcommand{\virginie}[1]{#1}
\newcommand{\jerome}[1]{#1}
\newcommand{\matthieu}[1]{#1}

\allowdisplaybreaks
\renewcommand{\H}{{\mathrm{H}}}
\newcommand{\w}{{\mathrm{w}}}
\newcommand{\mav}{{\mathrm{MAV}}}
\newcommand{\opt}{{\mathrm{opt}}}
\newcommand{\pav}{{\mathrm{PAV}}}
\newcommand{\calA}{{\mathcal{A}}}
\newcommand{\np}{{\mathrm{NP}}}
\newcommand{\cc}{{\mathrm{CC}}}
\newcommand{\naturals}{{\mathbb{N}}}
\newcommand{\reals}{{\mathbb{R}}}
\newcommand{\calR}{{\mathcal{R}}}
\newcommand{\fsc}{{f\text{-}\mathrm{sc}}}

\usepackage{natbib}
\usepackage{graphicx}
\usepackage[colorlinks,pagebackref=true,citecolor=green!40!black,linkcolor=red!60!black,urlcolor=blue!60!black]{hyperref}

\newtheorem{theorem}{Theorem}[section]
\newtheorem{proposition}{Proposition}[section]
\newtheorem{corollary}{Corollary}[theorem]
\newtheorem{lemma}[theorem]{Lemma}
\newtheorem{definition}{Definition}[section]
\newtheorem{example}{Example}
\definecolor{winered}{rgb}{0.6,0.1,0.1}
\newcommand{\myemph}[1]{{\color{winered}\emph{#1}}}

\hypersetup{
colorlinks=true,citecolor=blue!80!black,linkcolor=red!60!black,urlcolor=blue!80!black}

\theoremstyle{remark}
\newtheorem*{remark}{Remark}

\usepackage{bm}
\usepackage{xspace}
\usepackage{dsfont}

\newcommand{\expe}[1]{\mathbb{E}[#1]}

\usepackage{authblk}

\title{Online Approval Committee Elections\footnote{To appear in the Proceedings of IJCAI 2022.}}
\author[1,2]{Virginie Do}
\author[1]{Matthieu Hervouin}
\author[1]{Jérôme Lang}
\author[3]{Piotr Skowron}
\affil[1]{CNRS, LAMSADE, PSL, Université Paris Dauphine, France}
\affil[2]{Facebook AI Research}
\affil[3]{University of Warsaw, Poland}

\date{}

\begin{document}

\maketitle
\begin{abstract}
Assume $k$ candidates need to be selected.  The candidates appear over time. Each time one appears, it must be immediately selected or rejected---a decision that is made by a group of individuals through voting.
Assume the voters use approval ballots, i.e., for each candidate they only specify whether they consider it acceptable or not.
This setting can be seen as a voting variant of choosing $k$ secretaries.
Our contribution is twofold.
(1) We assess to what extent the committees that are computed online can proportionally represent the voters. (2) If a prior probability over candidate approvals is available, we show how to compute committees with maximal expected score.
\end{abstract}

\section{Introduction}\label{sec: 1}

In the vast majority of elections, the set of candidates is known upfront. Yet, there are contexts where candidates appear over time. A paradigmatic example is hiring for a job: candidates come every day to pass an interview, they are evaluated by members of a jury, {\em and then it must be decided immediately whether to hire them or not}.
When we must hire only one candidate and when the evaluation is performed by a single agent (human or algorithm), we get the classic {\em secretary problem}. The variant where several employees must be hired, is called the {\em multiple secretary problem}. These problems have numerous variants (for instance, depending on whether we want to optimize the rank or the value of the selected candidate, whether the distribution over candidates' values is known, etc.). However, when the candidates are evaluated by a set of voters, we obtain a {\em voting} version of the {\em secretary problem}, or equivalently, an {\em online} version of {\em multiwinner elections} (also called {\em committee elections}). 

A typical example is the recruitment of a team of researchers for a research unit, from a pool of candidates who are being interviewed one at a time. In this example, the jury member may have different backgrounds and may have heterogeneous preferences over candidates. A second example is another hiring scenario where the voters are explicit criteria (e.g., skills or demographic
attributes such as gender).

In generalized secretary problems, each candidate is usually evaluated by a single number \citep{babaioff2008online}.
In voting, numerical evaluations are not very convenient, and simpler types of ballots are typically used. In this paper we assume each voter submits an approval ballot. Deciding whether to approve or disapprove the current candidate is cognitively easy and the landscape of approval-based committee rules (ABC rules) is relatively well understood. In particular, ABC rules are well studied for \emph{proportionality}, which guarantees that voters' approvals are fairly represented in the elected committee. In contexts where voters are evaluation criteria, it ensures that these are well covered by the set of hired candidates.


Formally, we have a set of $n$ voters; at each time $t$, a new candidate $c_t$ is observed; $c_t$ is then approved or disapproved (possibly after being interviewed) by each voter; and we have to decide immediately whether to include $c_t$ in the committee. We have to select a committee of size $k$, and exactly $m$ candidates can be observed. 

If we could wait until all candidates have been interviewed, then we would be in the classic setting of approval-based committee elections. In this setting there is a plethora of well-understood rules with different distinctive properties \citep{lac-sko:abc-survey}; we would then pick one of these rules, say $f$, and compute its outcome. We cannot do this, because if candidate $c_t$ appears at time $t$, then it must be decided at time $t$ whether to hire it or not.\footnote{We could consider intermediate contexts where we can wait some amount of time before deciding to hire a candidate or not, but in this first study we will simply assume that the decision to hire a candidate or not must 
be done immediately and irrevocably. This is often realistic: a good candidate has good chances to find another job if not hired immediately.}  Still, the rule $f$ can serve as a reference: the set of candidates that would have been computed by $f$ if we had been able to wait can be considered the optimal set of winners, and can be used for measuring the quality of an online selection algorithm.

We consider two paradigms for evaluating rules.

First, we examine ABC rules $f$ that aim at maximizing certain scoring functions $\fsc$. Given a fixed function $\fsc$, we evaluate online rules by comparing the scores of returned committees with the scores of the optimal ones. We explain how the existing mechanisms for the multiple secretary problem can be applied in order to obtain rules that perform well.
Yet, our first contribution lies in designing optimal algorithms for the case, where a prior distribution over candidate approvals is available. 
The goal becomes to design algorithms that maximize the expected score of the selected committees. We show that for selected scoring functions $\fsc$ or under some assumptions, this can be done in polynomial time.

Second, we look at two {\em axioms of proportionality}. These properties require that each group of voters with cohesive preferences must be represented in the selected committee proportionally to its size. We give a polynomial-time computable online selection policy that satisfies the axiom of {\em proportional justified representation}~\citep{pjr17}. We further show that the stronger axiom of {\em extended justified representation} is not satisfiable in the online setting, but it can be approximated. We give two such approximation algorithms, one which gives the best possible approximation but is computationally inefficient, and the other one, which can be computed in polynomial time.


After discussing related work (Section \ref{sec:related}), we define the online committee elections problem (Section \ref{sec:preminimaries}). Then, we discuss the construction of policies maximizing the expected score (Section \ref{sec:welfare}), and the construction of policies providing some exact or approximate proportionality guarantees (Section \ref{sec:prop}). We conclude in Section \ref{sec:conclu}.

 \section{Related Work} \label{sec:related}

\paragraph{Online selection problems}
Our setting is close to generalized secretary problems \citep{babaioff2008online}, where the goal is to hire the best possible subset of candidates among a finite set of candidates arriving one at a time. A candidate's value is revealed upon arrival,
and the hiring decision must be taken immediately and cannot be changed afterwards. The connection between our work and this class of problems is detailed in Section \ref{sec:worstcase}. \citet{yu2019bi} consider a bi-criteria secretary problem with multiple choices, where each criterion can be seen as a voter, the subset of choices as the elected committee, and their objective function as the multiwinner approval voting objective. We consider instead a multiplicity of voters, and various objectives corresponding to general Thiele rules in multiwinner voting. In the single secretary problem variant of \citep{bearden2005multi}, there are multiple independent attributes, which can be seen as voters in a single-winner election. In the context of search engines, \cite{panigrahi2012online} aim to find a diverse set of items from an input stream, by maximizing a coverage function of multiple features. In our committee election framework, items can be seen as candidates and features as voters, yet their objective function is different from the committee scoring functions we consider, and proportional representation is not studied.

\paragraph{Social choice in online settings} Our work is mainly related to the study of proportional representation in committee elections, and in particular approval-based committee elections, which are surveyed in \citep{lac-sko:abc-survey}. In recent years, there has been increased interest in studying online versions of voting problems. Proportionality is studied in \citep{dey2017proportional} who formalize voting streams, a setting in which alternatives are fixed, but voters arrive in an online manner, which is the opposite to ours. In \citep{freeman2017fair}, the sets of voters and alternatives are fixed, but the valuations of each voter for an alternative varies over time. Utilities are defined at each timestep as the cumulative reward of each agent given past decisions, and the goal is to maximize Nash social welfare. \citet{lackner2020perpetual} consider a similar setting with ordinal preferences instead of cardinal valuations, and study voting rules that weight agents according to their past satisfaction. \citet{hossain2021fair} address the partial observability of voters' preferences. 

While these works study (repeated) single-winner elections, the only existing work on online multiwinner elections to our knowledge is \citep{oren2014online}. The difference with our setting is that they consider online random arrival of voters rather than candidates, and they do not study proportionality axioms. \citet{do2021online} also study a close online committee selection problem to ours, yet a major difference is the absence of voters in their case. Proportionality is then defined based on multiple demographic attributes and a distance to target proportions on these attributes. \virginie{In independent 
work, \citet{banerjee2022proportionally} study a similar setting of fair online allocation where each public good can be assimilated to an election candidate. While they focus on a quantitative notion of proportional fairness, we study welfare guarantees and qualitative proportionality axioms.}

Recent studies address fairness in online versions of other public decision-making problems, such as dynamic proportional rankings \citep{israel2021dynamic} and participatory budgeting \citep{lackner2021fairness}. 

\section{Preliminaries}
\label{sec:preminimaries}

For each $i \in \naturals$ we write $[i]$  to denote the set $\{1, \ldots, i\}$. By $\H(i)$ we denote the $i$-th harmonic number, i.e., $\H(i) = \sum_{j = 1}^i \nicefrac{1}{j}$. By $\w(\cdot)$ we denote the inverse function\footnote{We have $w(i) = \exp(W (\ln (i))$ where $W$ is the Lambert function, i.e. the inverse multivalued function of $x \mapsto x e^x.$} of $x \mapsto x^x$, i.e., $\w(i) = x$ if $x^x = i$; clearly $\w(i) = O(\log(i))$. Further, it holds that $\log(i) = O\left(\w(i)^2\right)$. 

An \myemph{approval-based election} (in short, an election) is a triple \jerome{$E = (C, N, k, (A_i)_{i \in N})$}, where $C = \{1, \ldots, m \}$ is the set of \myemph{candidates}, $N = \{1, 2, \ldots, n\}$ is the set of \myemph{voters}, $k$ is the desired size of the committee, \jerome{and for each $i \in N$,  $A(i) \subseteq C$ is the approval ballot associated with $i$, that is,} the set of candidates that $i$ finds acceptable. Conversely, we let $N(c) = \{i \in N \colon c \in A(i)\}$ denote the set of voters who approve candidate $c$.

We refer to $k$-elements subsets of $C$ as to size-$k$ \myemph{committees}. An \myemph{approval-based committee election rule} (in short, an ABC rule) is a function $\calR$ that takes as input an election $E = (C, N, k)$ and returns a nonempty set of committees; we call the elements of $\calR(E)$ winning committees. Typically we are interested in selecting a single winning committee, but we allow for ties.

An \myemph{online ABC rule} is an algorithm that iterates over the candidates according to the sequence $c_1, c_2, \ldots, c_m$, and in each step makes the decision whether to include a candidate at hand, $c_t$, in the winning committee, or not. When making such a decision we assume that the algorithm does not know the preferences of the voters over the candidates $c_{t'}$ with $t' > t$. In other words, we assume the candidates appear one after another over time. When a candidate $c_t$ appears, the voters' preferences regarding $c_t$ are revealed, and the algorithm needs to make an irrevocable decision of whether $c_t$ is selected or not.   

In the following sections, we evaluate winning committees $W \in \calR(E)$ based on the approval ballots $A(i)$ of the voters, either assigning a value to $W$ to measure aggregated satisfaction (Section \ref{sec:welfare}), or analyzing the proportionality axioms satisfied by $W$ (Section \ref{sec:prop}). Importantly, in the online setting we consider, the approval ballots are not available beforehand, but only once all candidates have been seen and approvals revealed. It is possible to analyze an ABC rule \emph{ex-ante}, by measuring the quality of $W$ in terms of $\mathbb{E}[f(|A(i) \cap W|)]$ for some $f$, as we do in Section \ref{sec:welfare}. It is also possible to evaluate it \emph{ex-post}, once all approval ballots are available and a committee has been elected, as we do in Section \ref{sec:prop}.

\section{Maximizing Aggregated Satisfaction} \label{sec:welfare}

In this section we look at the problem of maximizing the aggregated voters' satisfaction. Consider a voter $i$, who approves $r$ members of the elected committee~$W$, i.e., $r = |A(i) \cap W|$. Given a function $f\colon \naturals \to \reals$, we define the $f$-utility of $i$ from $W$ as $f(r)$. An \myemph{$f$-Thiele method}~\citep{thi:j:approval-OWA,lac-sko:abc-survey} is an ABC rule that maximizes the total $f$-utility of the voters: given an election $E = (C, N, k)$ it elects committees $W$ that maximize the following score:
\begin{align*}
\fsc(W) = \sum_{i \in N} f\left(|A(i) \cap W| \right) \text{.}
\end{align*}
Examples of Thiele methods commonly studied in the literature include
\begin{inparaenum}[(1)]
\item Multiwinner Approval Voting rule (MAV), with $f_{\mav}(r) = r$,
\item Proportional Approval Voting (PAV) with $f_{\pav}(r) = \H(r)$, and
\item Approval Chamberlin--Courant rule (CC), with $f_{\cc}(r) = \min(1, r)$.
\end{inparaenum}

\subsection{Unknown Distributions} \label{sec:worstcase}


For $\mav$, the problem of finding the best approximate committee can be casted as a \emph{multiple choice secretary problem}~\citep{kleinberg2005multiple}. This is because we can define the value of a candidate $c$ as $N(c)$, the number of voters approving $c$, and because $f_{\text{MAV}}$ is additive in the values. 

More generally, for all concave utility functions, we can directly apply the results for the \emph{submodular secretary problem} \citep{bateni2013submodular}. This includes not only $\mav$ but also CC and PAV, because $f_{\text{CC}}$ and $f_{\text{PAV}}$ are submodular set functions.
We obtain an online ABC rule with a constant-factor approximation guarantee, which works as follows. 
It first divides the sequence of candidates into $k$ roughly equal-size parts---the size of each part is between $\lfloor \nicefrac{m}{k} \rfloor$ and $\lceil \nicefrac{m}{k} \rceil$. From each part we select exactly one candidate as follows. Consider the $i$-th part, and assume the set of $(i-1)$ candidates $W_{i-1}$ has been already selected. To select the $i$-th candidate we first observe the first $\lceil \nicefrac{m}{ke} \rceil$ candidates in the $i$-th part of the sequence, and find one, call it $a_i$, that maximizes $f(W_{i-1} \cup \{a_i\})$. Next, we select the first candidate $c$ such that $f(W_{i-1} \cup \{c\}) \geq f(W_{i-1} \cup \{a_i\})$. If we found no such candidate in the $i$-th part of the sequence, we pick the last candidate from the $i$-th part, and move to the next part. This way, we select exactly one candidate from each part of the sequence. By the result of \cite{bateni2013submodular}, this algorithm returns a committee $W$ such that $\expe{\fsc(W)} \geq \frac{1 - \nicefrac{1}{e}}{7}\fsc(W_{\opt}) \approx 0.09\fsc(W_{\opt})$. 


The above results are valid for unknown distributions of values of (subset of) candidates, and do not depend on how these values are defined. The interest of the following Section \ref{sec:dynprog} is to leverage the specificity of our problem, where the objective function $\fsc$ depends on multiple voters and is additively decomposable in their utilities. Precisely, the input at each time $t$ is richer than the single value $N(c)$, and more specific than an abstract function $\fsc(S)$. Rather, we observe at each time the approvals of each voter for the arriving candidate $c_t$, i.e., $n$ multiple binary valuations. Using this specific problem structure, and an additional assumption of prior knowledge of approval probabilities, we show how dynamic programming can be used to go beyond the off-the-shelf solutions offered by generalized secretary problems.

\subsection{Known Distributions} \label{sec:dynprog}

We now assume that we have a known prior distribution: {\em we know the probability $p(i) = p(x \in A_i)$ that voter $i$ approves the next observed candidate $x$}. These probabilities may depend on $i$, and the events $x \in A_i$ and $x \in A_j$ for different $i \neq j$ need not be independent.


Whether it is realistic to assume we know $p(i)$ depends on the context. If we have a database of past instances on similar problems, then we can compute the approval frequency of voters (or of a given voter, if she appears in several instances and the database is not anonymous). 

For the sake of simplicity, we now assume that (UI) $p(i)$ has the same ({\em uniform}) value $p$ for all voters, and that the events that voter $i$ approves or not candidate $j$ are {\em independent}.
Assumption (UI) is classic in social choice. Note that (UI) is not necessary for some of our results to hold: what we need is only that the probability $P_j$ that a candidate is approved by $j$ voters is polynomial-time computable. Under (UI), $P_j = {n \choose j}p^j (1-p)^{n-j}$.

A policy is a function $\pi$ that decides, at each step when a new candidate comes and once the approvals and disapprovals for this candidate are observed, whether the candidate should be selected or not. More rigorously, the policy maps a {\em state} to a decision; we postpone the definition of a state because it varies with the rule used.

A history is a sequence of candidates together with associated votes and actions: $h = \langle (c_t, N(c_t), a_t), t = 1, \ldots, q\rangle$ for $q \leq m$. $c_t$ is the candidate observed at time $t$; $N(c_t)$ is the set of voters who approve $c_t$; and $a_t \in \{yes, no\}$ (decision to selecting or not $c_t$).
For instance, $h = \langle (a, \{1,3,4\}, yes), (b, \{1,2\}), no) \rangle$ is the history where $a$ is observed, approved by voters 1, 3 and 4, selected, then $b$ is observed, etc.
A history is {\em terminal} if either $q = m$ or the number of selected candidates is $k$. A policy is {\em safe} if all its induced histories select exactly $k$ candidates. Provided that $m \geq k$, safe policies exist. Each terminal history $h$ of a safe policy has an associated set of selected candidates $W(h)$ of cardinality $k$ and a reward $\fsc(W(h))$. A policy induces a probability distribution over histories, which in turns allows to define its expected score. An {\em optimal safe policy} is one with maximal expected score: 
\begin{align*}
    V^* = \max_{\pi} \mathbb{E}_{h \sim \pi}\left[ \fsc (W(h)) \right]
\end{align*}
We now show that we can express an optimal safe policy as a mapping from any state to an action $yes$ or $no$, and that it can be computed by a dynamic programming algorithm. The exact definition of a state varies with the multiwinner voting rule.
 A state is {\em full} if it corresponds to $k$ candidates having been selected already, and {\em tight} if its associated number of candidates already selected plus the number of candidates yet to be seen is equal to $k$.

\subsubsection{Multiwinner Approval Voting}

 
With $\mav$,  we have $f_{\text{MAV}}(W) = \sum_{c \in W} |N(c)|$. 
In this case a state is a triple  $s = (\alpha, \beta, \gamma)$,  where
\begin{itemize}
\item $\alpha \in \{1, \ldots, m\}$ is the number of candidates seen so far, including the currently observed candidate.
\item $\beta \in \{0, \ldots, \min(k,\alpha-1)\}$ is the number of candidates selected so far.
\item $\gamma \in \{0, \ldots, n\}$ is the number of voters who approve the current candidate.
\end{itemize}

The number of states is $(n+1)(k+1)\left(m+1 - \frac{k}{2}\right)$. 
A state is full if $\beta = k$ and tight if $\beta + m-\alpha +1 = k$. A safe policy must map every full state to $no$ and every tight state to $yes$. Note that for $\alpha = m$, a state obtained by following a safe policy is either tight or full.

Let $V^*(\alpha,\beta,\gamma)$ be the expected score of an optimal safe policy from $(\alpha,\beta,\gamma)$.

Under (UI), $V^*$ satisfies the Bellman equations
$$\begin{array}{l}
V(\alpha,\beta,\gamma, no) = \sum_{j=0}^n P_j V^*(\alpha+1,\beta,j)
\\
V(\alpha,\beta,\gamma, yes) = \gamma + \sum_{j=0}^n P_j V^*(\alpha+1,\beta+1,j)
\\
V^*(\alpha,\beta,\gamma) = \max\left( V(\alpha,\beta,\gamma, no), V(\alpha,\beta,\gamma, yes) \right)
\end{array}$$

$V(\alpha,\beta,\gamma, yes)$ and $V(\alpha,\beta,\gamma, no)$ are the expected utilities obtained when selecting (resp., not selecting) the current candidate in state $(\alpha,\beta,\gamma)$ and then following an optimal safe policy. (For a detailed justifications, see the Appendix). Thus, the optimal safe policy can be computed by dynamic programming by iterating on all states from $\alpha = m$ down to $\alpha = 1$. There are $O(n^2km)$ states and each state needs a summation over $n$ terms.


\begin{proposition}\label{prop:antemav}
For $\mav$, under assumption (UI), an optimal safe policy can be computed in time $O(n^2km)$.
\end{proposition}

\subsubsection{Chamberlin-Courant approval voting}

Recall that $f_{\text{CC}}(W) = |\{i \in N: W \cap A(i) \neq \emptyset\}|$.

\begin{proposition}\label{prop:antecc}
For CCAV, under assumption (UI), an optimal policy can be computed in time $O(n^3km)$.
\end{proposition}

The proof is similar to that of Proposition \ref{prop:antemav}, except that the state space is larger. The marginal value of a candidate $x$ when the current set of selected candidates is $S$ is 1 if some of the voters who have no approved candidate in $S$ approves $x$, and 0 otherwise. This means that to be able to determine the marginal value of a new candidate, it is necessary to know the number of voters, among those who have disapproved all candidates selected so far, who approve the currently observed candidate. So now a state is a tuple $s = (\alpha, \beta,\gamma, \delta)$, where $\alpha$ and $\beta$ are as before, $\delta \in \{0, \ldots, n\}$ is the number of all ``unsatisfied'' voters, i.e., those that approve no candidate in the current selection, and $\gamma \in \{0, \ldots, \delta\}$ is the number of those unsatisfied voters who approve the current candidate.

The optimal policy is computed by dynamic programming, iterating on all states, with the Bellman equations
$$\begin{array}{l}
V(\alpha,\beta,\gamma,\delta, no) = \sum_{i=0}^\delta p(i,\delta) V^*(\alpha+1,\beta,i,\delta)\\
V(\alpha,\beta,\gamma, \delta, yes) = \\
 \qquad \gamma + \sum_{i=0}^{\delta - \gamma} p(i,\delta-\gamma) V^*(\alpha+1,\beta+1,i,\delta-\gamma) \\
V^*(\alpha,\beta,\gamma,\delta) = \max\left( V(\alpha,\beta,\gamma,\delta, no), V(\alpha,\beta,\gamma, \delta,yes) \right)
\end{array}$$
where $p(i,\delta) = {\delta \choose i}p^i (1-p)^{\delta-i}$.
Now there are $\Theta(n^2km)$ states, therefore the algorithm runs in $O(n^3km)$.

\subsubsection{PAV and general Thiele rules}

The states used for CCAV are no longer sufficient: to know the marginal gain for voter $i$, of the current candidate, relatively to the current selection, we must store the number of candidates already selected approved by $i$.

The number of states is now in the order of $mk^n$. The dynamic programming algorithm still works but runs in time exponential in $n$. Consequently, we get polynomial-time computability if the number of voters is constant.

\begin{theorem}\label{prop:antepav}
For all Thiele rules, including PAV, $f_{CCAV}$, if the number of voters is constant then an optimal safe policy can be computed in polynomial time.
\end{theorem}

Small values of $n$ are realistic: we can think of a small jury, or of the interpretation of voters as criteria.






A subclass of Thiele rules for which the problem is still tractable consists of rules
such that the number of values of the score vector is bounded by a constant. This is the case for MAV and CCAV, 
but also for other rules such as truncated $\pav$, defined by the vector $(1, \nicefrac12, 0, \ldots, 0)$.

\section{Proportionality} \label{sec:prop}

In this section we focus on the concept of proportionality. Our goal is to design online ABC rules which would guarantee each minority of the voters the right to decide about a part of the elected committee.

Recall that $k$ is the committee size.
For an integer $\ell \in [k]$ we say that a group of voters $S \subseteq N$ is \myemph{$\ell$-cohesive} if \begin{inparaenum}[(1)]
\item it is large enough, $|S| \geq \ell \cdot \nicefrac{n}{k}$, and
\item its members approve of at least $\ell$ common candidates, $|\bigcap_{i \in S} A(i)| \geq \ell$.
\end{inparaenum}
We extend this notion, and define an approximate variant of $\ell$-cohesiveness. Given an $\alpha > 1$ we say that a group $S$ is $\alpha$-$\ell$-cohesive if 
\begin{inparaenum}[(1)]
\item $|S| \geq \alpha \cdot \ell \cdot \nicefrac{n}{k}$, and
\item $|\bigcap_{i \in S} A(i)| \geq \ell$.
\end{inparaenum}


The two notions of proportionality that are commonly considered in the literature are proportional justified representation (PJR)~\citep{pjr17}, and extended justified representation (EJR)~\citep{justifiedRepresentation}. Below we define their approximate variants. 

\begin{definition}[Proportional justified representation]
Given an $\alpha > 1$ we say that a committee $W$ satisfies an $\alpha$-Proportional Justified Representation ($\alpha$-PJR) if for each $\ell \in [k]$ and each $\alpha$-$\ell$-cohesive group of voters $S$ it holds that $|\bigcup_{i \in S} A(i) \cap W| \geq \ell$.
\end{definition}

Analogously, we define the axiom of $\alpha$-EJR.

\begin{definition}[Extended justified representation]
Given an $\alpha > 1$ we say that a committee $W$ satisfies an $\alpha$-Extended Justified Representation ($\alpha$-EJR) if for each $\ell \in [k]$ and each $\alpha$-$\ell$-cohesive group of voters $S$ there exists a voter $i \in S$ who approves of at least $\ell$ committee members, i.e.,  $|A(i) \cap W| \geq \ell$.
\end{definition}

We say that a committee election rule satisfies $\alpha$-PJR if each committee returned by the rule satisfies $\alpha$-PJR. Analogously, we define what it means that a rule satisfies $\alpha$-EJR. These axioms form a hierarchy: if a rule satisfies $\alpha$-EJR then it also satisfies $\alpha$-PJR.
If a rule satisfies $\alpha$-EJR (respectively, $\alpha$-PJR) for $\alpha = 1$ then we simply say that the rule satisfies EJR (respectively, PJR). EJR is a very strong axiom and for the time being it is known to be satisfied only by PAV~\citep{justifiedRepresentation} and Rule X~\citep{pet-sko:laminar}. Further, Sequential Phragm\'en's Rule satisfies $2$-EJR~\citep{skowron:prop-degree}.

\subsection{Proportional Justified Representation}

Somehow surprisingly, it appears that the axiom of PJR can be satisfied in the online setting by the following Greedy Budgeting Rule. Each voter is initially given $1$ dollar. When a candidate $c \in C$ arrives we look if the voters who approve~$c$ have at least $\nicefrac{n}{k}$ dollars in total. If so, we add $c$ to the committee and ask the voters from $N(c)$ to pay $\nicefrac{n}{k}$. The properties of the algorithm do not depend on how spread the cost of $\nicefrac{n}{k}$ among the voters from $N(c)$, but a fair policy would suggest to do it as evenly as possible. This way the rule would resemble the method of equal shares~\citep{pet-sko:laminar,pierczynski2021proportional}.

Since the voters have in total $n$ dollars, and buying each candidate costs $\nicefrac{n}{k}$, it is clear that the rule cannot select more than $k$ candidates. If it picks less, we can add the last candidates that appear, so that exactly $k$ of them are selected.

\begin{theorem}
The Greedy Budgeting rule satisfies PJR.
\end{theorem}
\begin{proof}
Consider an election $E = (C, N, k)$, and towards a contradiction suppose the committee $W$ returned by the Greedy Budgeting Rule does not satisfy PJR. Let $S$ be an $\ell$-cohesive group such that $|\bigcup_{i \in S} A(i) \cap W| < \ell$.

Each time we select a candidate, we ask the voters to pay exactly $\nicefrac{n}{k}$. Since $|\bigcup_{i \in S} A(i) \cap W| \leq \ell-1$ we asked the voters from $S$ to pay at most $(\ell-1) \cdot \nicefrac{n}{k}$. Since $|S| \geq \ell \cdot \nicefrac{n}{k}$ we get that the voters from $S$ have at least $\nicefrac{n}{k}$ dollars at each step of the rule. Consequently, each time when a candidate from $\bigcap_{i \in S} A(i)$ appears, these voters have enough money to buy it.
As a result, each candidate from $\bigcap_{i \in S} A(i)$ would be selected. There are at least $\ell$ such candidates. This gives a contradiction and completes the proof.
\end{proof}




\subsection{An Online Algorithm with $\H(k)$-EJR}\label{sec:ogca}

We now move to the case of extended justified representation (EJR). We start by defining the Online Greedy Cohesive algorithm (OGCA), and next we will prove that OGCA satisfies $\H(k)$-EJR \matthieu{(note that OGCA favors large groups of voters and does not satisfy JR)}.

Consider a candidate $c \in C$ that arrives. \matthieu{If there exists a group of voters $S \subseteq N(c)$ with $|S| \geq \H(k) \cdot \ell\cdot\nicefrac{n}{k}$ such that each voter from $S$ approves less than $\ell$ candidates selected so far, then OGCA accepts~$c$}. Otherwise, $c$ is rejected. If the rule were to select less than $k$ candidates, the candidates that arrived last are accepted so that the committee seats are filled.

\begin{theorem}\label{thm:logkejr}
For each election $E$, OGCA returns a size-$k$ committee that satisfies $\H(k)$-EJR.
\end{theorem}
\begin{proof}
The fact that the algorithm satisfies $\H(k)$-EJR follows directly from its definition, we will prove that it selects at most $k$ candidates using a budgeting argument.

With each candidate $c$ we associate the price of $\nicefrac{n}{k}$.
When the algorithm accepts $c$, its cost is spread equally among the voters who approve it. Notice that for any $\ell \in [k]$, each voter buys at most $\ell$ candidates forming an $H(k)$-$\ell$-cohesive group. For any such candidate $c,$ $\matthieu{|N(c)|} \geq \H(k) \cdot \ell \cdot\nicefrac{n}{k}.$ 

For $\ell=1$, each voter buys at most one candidate $c$ forming an $H(k)$-$1$-cohesive group, and in such a case the voter pays at most $\nicefrac{n}{k} \cdot \frac{1}{H(k) \cdot 1\cdot \nicefrac{n}{k}} = \frac{1}{H(k)}$. 

For $\ell=2$, a voter buys at most two candidates forming an $H(k)$-$2$-cohesive group. One of such candidates could have been bought before (as a candidate forming a $H(k)$-$1$-cohesive group), and the voter would pay $\frac{1}{H(k)}$ for it. For the second candidate, the voter would pay at most $\nicefrac{n}{k} \cdot \frac{1}{H(k) \cdot 2 \cdot \nicefrac{n}{k}} = \frac{1}{2H(k)}$. 

Repeating the reasoning for $\ell=1,\ldots, k$, we get that each voter paid at most: $\sum_{\ell=1}^k \nicefrac{1}{\ell H(k)} = 1.$ Thus, the total amount of money paid is at most equal to $n$. Since each candidate costs $\frac{n}{k},$ our algorithm could have selected at most $k$ candidates. 
\end{proof}

Interestingly, in terms of proportionality guarantees the Online Greedy Cohesive algorithm is optimal.

\begin{theorem}
For each $\epsilon > 0$ there exists no online ABC rule that would satisfy $(1 - \epsilon)\H(k)$-EJR.
\end{theorem}
\begin{proof}
For the sake of contradiction assume that there exists an algorithm $\calA$ that satisfies $(1- \epsilon)\H(k)$-EJR for some rational $\epsilon > 0$. 

Let us fix $k$, and assume the number of voters $n$ is such that $(1- \epsilon)\cdot \nicefrac{n}{k}$ is an integer. In the first round there arrive candidates who are approved by $(1- \epsilon)\H(k) \cdot \nicefrac{n}{k}$ voters. Each such candidate is approved by a disjoint group of voters. Assume that the number of such candidates equals $m_1$, where: $m_1 = \left\lfloor\frac{n}{(1- \epsilon)\H(k) \cdot \nicefrac{n}{k}}\right\rfloor \geq \frac{k}{(1- \epsilon)\H(k)} - 1 \text{.}$ Note that each such a candidate must be selected by $\calA$. Indeed, if one of them were not selected, the algorithm could violate $(1- \epsilon)\H(k)$-EJR. This could happen, for example, if all the remaining candidates that have not yet arrived were approved by no voters. 

In the second round there arrive $m_2$ candidates, each approved by a different group of $2(1- \epsilon)\H(k) \cdot \nicefrac{n}{k}$ voters, where: $m_2 \geq \nicefrac{k}{2(1- \epsilon)\H(k)} - 1 \text{.}$
By the same argument as before, we infer that $\calA$ must accept each such a candidate. 

Analogously, in the $i$-th round, $i \leq \frac{k}{(1- \epsilon)\H(k)}$ there arrive $m_i$ candidates: $m_i \geq \frac{k}{i(1- \epsilon)\H(k)} - 1 \text{,}$
and each of them must be accepted by $\calA$. In total $\calA$ must have accepted the following number of candidates:
\begin{align*}
m = \sum\limits_{i = 1}^{\left\lfloor \frac{k}{(1- \epsilon)\H(k)} \right\rfloor} m_i \text{.}
\end{align*}
In the following sequence of estimations we use the fact that for each $i$ it holds that $ \log(i) \leq \H(i) \leq \log(i) +2$:
\begin{align*}
    m &= \sum_{i = 1}^{\left\lfloor \frac{k}{(1- \epsilon)\H(k)} \right\rfloor} m_i \geq \sum_{i = 1}^{\left\lfloor \frac{k}{(1- \epsilon)\H(k)} \right\rfloor} \left(\frac{k}{i(1- \epsilon)\H(k)} - 1\right) \\
    &\geq \frac{k}{(1- \epsilon)\H(k)} \cdot \H\left(\left\lfloor \frac{k}{(1- \epsilon)\H(k)} \right\rfloor\right) - \left\lfloor \frac{k}{(1- \epsilon)\H(k)} \right\rfloor \\
    &\geq \frac{k}{(1- \epsilon)\H(k)} \cdot \left(\log\left(\frac{k}{(1- \epsilon)\H(k)} \right) - 2 \right) \\
    &\geq \frac{k}{(1- \epsilon)\H(k)} \cdot \Big(\H(k) -\log\big((1- \epsilon)\H(k)\big) - 4\Big) \\
    &\geq \frac{k}{(1- \epsilon)}\left(1 - \frac{\log\big(\H(k)\big) + 4}{\H(k)} \right)\text{.}
\end{align*}
Note that $\lim_{k \to \infty} \frac{\log\big(\H(k)\big) + 4}{\H(k)} = 0$, thus for sufficiently large $k$ the last expression in the sequence of inequalities is larger than $k$. We infer that $\calA$ would need to select more than $k$ candidates, a contradiction. 
\end{proof}


\subsection{A Polynomial-Time Algorithm Satisfying $\w(k)^2$-EJR}

The OGCA algorithm presented in Section \ref{sec:ogca} cannot be computed in a polynomial time. This is because checking if there exists an $\ell$-cohesive group is $\np$-hard~\citep{proprank}. In this section we define an algorithm that runs in polynomial time, and which offers only a slightly worse EJR guarantee than OGCA.

Our algorithm, which we call Subcommittees via Greedy Budgeting Rule (SGBR), is defined as follows. Let $\alpha = \lceil\w(k)\rceil$. The idea is to independently elect $\alpha$ smaller committees, each of size $k' = \lfloor \nicefrac{k}{\alpha} \rfloor$. We elect the $i$-th subcommittee, $i \in [\alpha]$, using the Greedy Budgeting Rule, but with a constraint that we can pick only the candidates who are approved by at least $\frac{n \alpha^i}{k}$ voters.

Formally, we assume that each voter is given an initial budget of $(1,\ldots, 1) \in [0,1]^{\alpha}$, that is $\alpha$ independent coins. The $i$-th coin can be used for buying the candidates who are approved by at least $\frac{n \alpha^i}{k}$ voters. Each candidate costs $\frac{n\alpha}{k}$ coins. When a candidate $c \in C$ arrives, we find the largest pair $i \in N$ and $S \subseteq N(c)$ (we first maximize $i$, and second $|S|$) such that:  
\begin{inparaenum}[(1)]
\item $|S| \geq \frac{n \alpha^i}{k}$,
\item Each voter from $S$ has at least $\frac{n\alpha}{k |S|}$ coins of type $i$ left. That is, those voters can afford to buy candidate $c$ assuming each of them paid with the coins of type $i$, and each would pay the same amount of money.
\end{inparaenum}
If such pair $(i, S)$ does not exist, we reject $c$. Otherwise, $c$ is accepted and we ask each voter from $S$ to pay $\frac{n\alpha}{k |S|}$ for $c$.

Since each voter has in total $\alpha$ coins, and buying each candidate costs $\frac{n\alpha}{k}$ the algorithm selects $\leq k$ candidates.  



\begin{theorem}
SGBR satisfies $\lceil\w(k)\rceil^2$-EJR.
\end{theorem}

\begin{proof}
For the sake of contradiction assume that given an election $E = (C, N, k)$ Subcommittees via Greedy Budgeting returns a committee $W$ that fails $\alpha^2$-EJR. Let $S$ be a subset of $\alpha^2$-$\ell$-cohesive voters such that $|S| = \alpha^2 \ell \frac{n}{k}$ and $|\cap_{i \in S} A(i)|\geq \ell$ for some $\ell \in [k]$, and that for all $v \in S$, we have $|A(v)\cap W|<\ell$.

There exists $j \in [\alpha]$ such that $\frac{n}{k}\alpha^j \leq |S| \leq \frac{n}{k}\alpha^{j+1}$. From that it follows that $\ell \leq \alpha^{j-1}$. Let $W_j$ be the $j$-th subcommittee. For each elected candidate from $W_j$ a single voter can pay at most:
\begin{align*}
\frac{\frac{n\alpha}{k}}{\frac{n}{k}\alpha^j}= \frac{1}{\alpha^{j-1}}\leq \frac{1}{\ell}.
\end{align*}

Since each voter from $S$ approves at most $(\ell-1)$ candidates in $W_j$, they paid at most $(\ell-1)\frac{1}{\ell}$ and their remaining budget is greater than $\frac{1}{\ell}\geq \frac{1}{\alpha^{j-1}}$. Thus, when a candidate from $\cap_{i\in S} A(i) \backslash W_j$ appears there are at least $|S|$ voters, $|S| \geq \frac{n}{k} \alpha^{j}$, each having at least $\frac{1}{\alpha^{j-1}}$ coins of type $j$ left. Thus, their total budget is sufficient to buy the candidate: $|S| \cdot \nicefrac{1}{\alpha^{j-1}} \geq \nicefrac{n \alpha^j}{k} \cdot \nicefrac{1}{\alpha^{j-1}} = \nicefrac{n \alpha}{k}\,\text{.} $

Consequently, each candidate from $\cap_{i\in S} A(i)$ would be selected. There are at least $\ell$ such candidates, thus $|A(v)\cap W| \geq \ell$ for each $v \in S$. This gives a contradiction and completes the proof.
\end{proof}

\begin{remark}
While EJR always implies PJR, $\alpha$-EJR does not necessarily imply PJR. Still, $\alpha$-EJR provides desirable guarantees even when PJR is not satisfied. PJR is a property that is often used both in the context of proportionality and of diversity. If our primary focus is to provide voters from large cohesive groups with multiple representatives, then $\alpha$-EJR should be considered. For example, assume there are $n$ voters, $2n$ candidates and $k = n$ is the desired committee size. Assume that each of the first $n$ candidates that arrive is approved by exactly one voter, and each such candidate is approved by a different voter. Further, assume that each of the next $n$ candidates is approved by all $n$ voters. Then the committee that consists of the first $n$ candidates satisfies PJR (in fact, this would be the committee selected by the Greedy Budgeting Rule). Such a committee would be bad from the perspective of $\alpha$-EJR, since each voter would have only a single representative. In this case, our algorithms for $\alpha$-EJR would return clearly better committees. 
\end{remark}

\section{Conclusion}\label{sec:conclu}

Our main message is that online approval-based committee elections are easier than we might have thought. First, for Thiele rules with submodular score functions we have a constant-factor approximation algorithm (obtained as a corollary of a known result), and a dynamic programming algorithm for maximizing the expected score, which runs in polynomial time for some rules or under some assumptions.
Second, we have an algorithm that returns a committee satisfying PJR, and two that return an approximate EJR committee: one with an optimal ratio, the other one running in polynomial-time. 

A further work direction consists in moving to ordinal preferences, where voters rank the current candidates in the list of candidates already observed. 

\newpage

\section*{Acknowledgements} This work was funded in part by the French government under management of Agence Nationale de la Recherche as part of the "Investissements d'avenir" program, reference ANR-19-P3IA-0001 (PRAIRIE 3IA Institute). We also thank the IJCAI-22 anonymous reviewers for useful comments and suggestions.

\bibliography{references}

\newpage
\appendix

\section[A]{Details for Section 4.2}

{\bf Proposition 4.1} {\em For $\mav$, an optimal safe policy can be computed in time $O(n^2km)$.}

\begin{proof}
 We show that we can express an optimal safe policy as a mapping from any state to an action $yes$ or $no$, and that it can be computed by the following dynamic programming algorithm.
 Let $P_j = {n \choose j}p^j (1-p)^{n-j}$.
 \medskip
 
Recall that a state is a triple  $s = (\alpha, \beta, \gamma)$,  where
\begin{itemize}
\item $\alpha \in \{1, \ldots, m\}$ is the number of candidates seen so far, including the currently observed candidate.
\item $\beta \in \{0, \ldots, \min(k,\alpha-1)\}$ is the number of candidates selected so far.
\item $\gamma \in \{0, \ldots, n\}$ is the number of voters who approve the current candidate.
\end{itemize}

The number of states is $(n+1)(k+1)\left(m+1 - \frac{k}{2}\right)$. 
A state is full if $\beta = k$ and tight if $\beta + m-\alpha +1 = k$. A safe policy must map every full state to $no$ and every tight state to $yes$. Note that for $\alpha = m$, a state obtained by following a safe policy is either tight or full. \medskip

We prove that Algorithm \ref{alg:opt-sav}  computes an optimal safe policy.


\begin{algorithm}\label{alg:opt-sav}
\caption{Optimal policy, MAV}
\begin{algorithmic}
\FOR{$\gamma \in 0, \ldots n$} 
\STATE ($\star$1)  $\pi(m,k,\gamma) = no$; $V^*(m,k,\gamma) = 0$
\STATE ($\star$2) $\pi(m,k-1,\gamma) = yes$; $V^*(m,k-1,\gamma) = \gamma$
\ENDFOR
\FOR{$\alpha = m-1, m-2, \ldots, 1$} 
\FOR{$\beta = 0, \ldots, \min(k,\alpha-1)$} 
\FOR{$\gamma = 0, \ldots, n$} 
\IF{$\beta + m-\alpha +1 = k$} \STATE ($\star$3) $\pi(\alpha,\beta,\gamma) = yes$
\ELSE
\IF{$\beta = k$} \STATE ($\star$4) $\pi(\alpha,\beta,\gamma) = no$
\ELSE
\STATE ($\star$5) $V(\alpha,\beta,\gamma, no) = \sum_{j=0}^n P_j V^*(\alpha+1,\beta,j)$
\STATE ($\star$6) $V(\alpha,\beta,\gamma, yes) = \gamma + \sum_{j=0}^n P_j V^*(\alpha+1,\beta+1,j)$
\STATE $V^*(\alpha,\beta,\gamma) = \max\left( V(\alpha,\beta,\gamma, no), V(\alpha,\beta,\gamma, yes) \right)$
\STATE ($\star$7) $\pi(\alpha,\beta,\gamma) = argmax_{a \in \{yes, no\}} V(\alpha,\beta,\gamma,a)$
\ENDIF
\ENDIF
\ENDFOR 
\ENDFOR 
\ENDFOR 
\RETURN $\pi$
\end{algorithmic}
\end{algorithm}
\medskip

Let $S_{\alpha \rightarrow m} = \{(\alpha', \beta, \gamma) \in S: \alpha' \geq \alpha\}$. A {\em $\alpha$-partial policy} $\pi_{\alpha \rightarrow m}$ maps every state in $S_{\alpha \rightarrow m}$ to an action. 

($\star$1), ($\star$2), ($\star$3) and ($\star$4) ensure that the returned policy is safe: whenever the number of remaining candidates (including the current one) $m-\alpha+1$ plus the number of selected candidates $\beta$ is equal to $k$, all remaining candidates are selected, and whenever we have already selected $k$ candidates, we never select another one. This ensures that $k-m+\alpha-1 \leq \beta \leq k$ holds at any stage of the execution of the algorithm (in particular, for $\alpha = m$, the only possible values of $\beta$ are $k-1$ and $k$).

We now prove by backward induction on $\alpha$ that (H) for each state $(\alpha, \beta, \gamma)$, the partial policy $\pi_{\alpha \rightarrow m}$ defined as the restriction of $\pi$ (output of the algorithm) to  $S_{\alpha \rightarrow m}$, is optimal among all $\alpha$-partial policies, and therefore $V^*(\alpha,\beta,\gamma)$ is the optimal score of a safe $\alpha$-partial policy.

When $\alpha = m$, this is obvious: if we have selected $k$ candidates already, the only possible action is $no$; and if we have selected $k-1$ candidates, the last candidate in the pool must be selected (and the reward is $\gamma$).

Assume now that (H) holds for $m, \ldots, \alpha+1$. 
If we choose not to select the current candidate in state $(\alpha,\beta,\gamma)$, and follow $\pi$ from $\alpha+1$ to $m$, then ($\star$5) the expected score of the output policy will be $\sum_{i=0}^n P_i V^*(\alpha+1,\beta,i)$, since at stage $\alpha+1$, there are still $\beta$ candidates selected, and the probability that the next state is $(\alpha+1,\beta+1,j)$ is $P_j$.


If we choose to select the current candidate in state $(\alpha,\beta,\gamma)$, and follow $\pi$ from $\alpha+1$ to $m$, then ($\star$6) the expected score of the output policy will be $\gamma + \sum_{j=0}^n P_j V^*(\alpha+1,\beta+1,j)$: because $\fsc_{MAV}$ is additive, the $\gamma$ approvals for the selected candidate contribute $\gamma$ to the final score of the policy, and the next state will be $(\alpha+1,\beta+1,j)$ with probability $P_j$.

Finally, ($\star 7$), together with the fact that $\pi_{\alpha+1 \rightarrow m}$ is optimal, ensures that $\pi_{\alpha \rightarrow m}$ is optimal, that is, H holds for $\alpha$. This ensures the optimality of the output policy.

Finally, the algorithm iterates for all states, and for each state, $(\star 4)$ and $(\star 5)$ take $O(n)$ operations, therefore the algorithm runs in $O(n^2km)$.
\end{proof}




\begin{example}
Let $m = 4$, $k = 2$, $n = 3$, $p_1 = p_2 = p_3 = \frac12$ and assume votes are independent, thus $P_0 = P_3 = \frac18, P_1 = P_2 = \frac38$. The optimal policy and the values $V^*(.)$ are depicted on the table below.

\begin{center}
\begin{tabular}{c|c|c|c|c}
& $\gamma = 0$ &$\gamma = 1$ & $\gamma = 2$ & $\gamma = 3$ \\ \hline
$(\alpha,\beta) = (4,2)$ & no, $0$ & no, $0$ & no, $0$ & no, $0$\\
$(\alpha,\beta) = (4,1)$ & yes, $0$ & yes, $1$ & yes, $2$ & yes, $3$ \\
$(\alpha,\beta) = (3,2)$ & no, $0$ & no, $0$ & no, $0$ & no, $0$\\
$(\alpha,\beta) = (3,1)$ & no, $\nicefrac32$ & no, $\nicefrac32$ & yes, $2$ & yes, $3$\\
$(\alpha,\beta) = (3,0)$ & yes, $\nicefrac32$ & yes, $\nicefrac52$ & yes, $\nicefrac72$ & yes, $\nicefrac92$\\
$(\alpha,\beta) = (2,1)$ & no, $\nicefrac{15}{8}$ &  no, $\nicefrac{15}{8}$ & yes, $2$ & yes, $3$\\
$(\alpha,\beta) = (2,0)$ & yes, $\nicefrac{21}{8}$ &yes, $\nicefrac{23}{8}$ &yes, $\nicefrac{31}{8}$ &yes, $\nicefrac{39}{8}$  \\
$(\alpha,\beta) = (1,0)$ & no, $\nicefrac{75}{16}$ & no, $\nicefrac{75}{16}$& no, $\nicefrac{75}{16}$& yes, $\nicefrac{81}{16}$  \\
\end{tabular}
\end{center}
Assume the first observed candidate $x_1$ receives two approvals: $(\alpha, \beta,\gamma) = (1,0,2)$, $x_1$ is not selected. Now the second observed candidate $x_2$ receives two approvals as well: $(\alpha, \beta,\gamma) = (2,0,2)$, $x_2$ is selected. Now the third observed candidate $x_3$ receives one approval as well: $(\alpha, \beta,\gamma) = (3,1,1)$, $x_3$ is not selected, and the last candidate $x_4$ is selected no matter how many approvals it receives.  Assume it receives one approval, then the output committee $\{x_2,x_4\}$ has total score 3 whereas the optimal committee (evaluated offline) is $\{x_1,x_2\}$, with score 4.
\end{example}

The algorithm runs in polynomial time as long as the $P_j$'s can be computed in polynomial time. This is the case for our assumption but also if, for instance, the voters are associated with different approval probabilities $p_i$, but can be clustered in a fixed number of types $\{t_j: j = 1, \ldots, Q\}$, each with a specific $p_j$.\medskip


\end{document}